\documentclass[11pt]{article}

\usepackage{amsmath,amssymb}
\usepackage{amsthm} 
\usepackage{newtxtext}
\usepackage{newtxmath}
\usepackage[margin=1in]{geometry}
\theoremstyle{definition}

\theoremstyle{plain}
\newtheorem{theorem}{Theorem}

\usepackage{tikz}
\usetikzlibrary{arrows.meta, positioning}
\newcommand{\wed}{\wedge^{2}}
\newcommand{\Gr}{\mathrm{Gr}}

\begin{document}

\begin{center}
\textbf{\Large Some intuition for why cooperative systems ``look 1-dimensional''\\
\medskip
  and 2-cooperative systems ``look 2-dimensional''}

\medskip
Eduardo Sontag, Northeastern University, July 2026

\end{center}

\medskip
\begin{center}
\begin{minipage}{0.86\textwidth}
\small
\textbf{Abstract.} It is known that cooperative systems, and more
generally systems monotone with respect to cones, behave under
appropriate irreducibility conditions like one-dimensional systems, as
evidenced by results such as Hirsch's generic convergence theorem. It
is also known, following work of S\'anchez and others, that
two-cooperative systems, those preserving a generally nonconvex cone of
rank two (tested through the diminishing of sign
variations), behave like two-dimensional systems, as evidenced by
Poincar\'e--Bendixson-type theorems. In these notes I attempt to give
some geometric intuition for these dimensionality reductions, based on
Birkhoff--Hilbert contractions of the projective metric on a positive
cone.
\end{minipage}
\end{center}
\medskip

These notes began as a way for me to understand intuitively why
cooperative systems behave like 1-dimensional systems and
2-cooperative systems behave, in many ways, like 2-dimensional systems.
The theory justifying these assertions goes back to the work of Hirsch
and of Hal Smith on monotone systems \cite{HirschSmith2005, Smith1995},
as well as work of S\'anchez and many others on $k$-monotone systems
and, in particular, versions of the Poincar\'e--Bendixson theorem
\cite{Smith1979, Sanchez2009, FengWangWu2021, FengWangWu2022,
WeissMargaliot2021, KatzGiordanoMargaliot2025}. A closely related line,
developed by Forni and Sepulchre and by Mostajeran and Sepulchre
\cite{ForniSepulchre2016, MostajeranSepulchre2018}, casts these ideas
geometrically as \emph{differential positivity}: the linearised flow
infinitesimally contracts a cone field of rank $k$, and the
$k$-dimensional dominant distribution shapes $k$-dimensional integral
submanifold attractors. My emphasis here is different, and more
elementary. Rather than the differential-geometric machinery of cone
fields and integral manifolds, I take the Birkhoff--Hilbert
projective-metric contraction \cite{Birkhoff1957, Birkhoff1967, Hopf1963}
as the mechanism: strict positivity does not merely map a cone inside
itself, it contracts an explicit metric at a definite rate. This yields
a quantitative, finite-time statement of the dimensional collapse, with
no compactness, integrability, or attractor hypothesis, and I think it
is one alternative way to provide the intuition. Having written the
notes, I decided they were worth sharing: the point of view seems to me
an illuminating one. Whether the viewpoint itself is novel I do not
know; I have simply not found this approach elsewhere.

To set the stage, I will first discuss standard cooperative systems.
In that case, generic variations of a given initial state are
approximately mapped, at each time $\tau$, into the line spanned by
the tangent to the trajectory starting from that state. A parallel
discussion for 2-cooperative systems then follows.

A few classical tools recur throughout. The main one is \emph{Birkhoff's theorem} \cite{Birkhoff1957, Birkhoff1967, Hopf1963}, that a linear
map sending a cone into its interior contracts the Hilbert projective
metric on that cone; this is the source of every ``alignment''
statement below. Alongside it we use the \emph{Perron--Frobenius}
theorem for Metzler (essentially nonnegative) matrices, \emph{compound
matrices} and exterior powers to move from vectors to 2-planes, and, in
the final section, the language of \emph{exponential separation}
(dominated splitting) to state precisely what the alignment amounts to.

The rest of the note is organized as follows. The present section fixes
definitions and assumptions and records the contraction rate we rely
on. Section~1 treats the 1-positive (cooperative) case and the sense in
which the dynamics look one-dimensional; Section~2 treats the 2-positive
case and the two-dimensional envelope. Section~3 states the precise
skeleton behind both, and an appendix supplies the explicit
finite-window estimate underlying the contraction rate.

Although the machinery is developed for general $k$, we emphasize
$k=2$ because strongly 2-cooperative systems satisfy a
Poincar\'e--Bendixson-type property: under the standard compactness and
strong-positivity hypotheses, a compact $\omega$-limit set containing
no equilibrium is a periodic orbit \cite{Smith1979, Sanchez2009, FengWangWu2021, FengWangWu2022}.
This provides one motivation for saying that their asymptotic geometry
can look two-dimensional, and it constrains the structure of compact
attractors \cite{MargaliotWuSontag2025}. The alignment argument
developed here is not itself a proof that every bounded trajectory
converges to either an equilibrium or a periodic orbit; it is an
account of why the tangent
geometry along a trajectory collapses toward a dominant
two-dimensional subspace.

\textbf{Standing definitions and assumptions}

Throughout, we consider a differential equation system $\dot x = f(x)$
evolving in an open subset of Euclidean space, with $f$ at least
continuously differentiable, and assume that solutions of interest are
defined for all $t\geq0$.

Fix a reference
trajectory $x(t) = \varphi_t(x_0)$. A system $\dot x = f(x)$ is
\emph{$k$-cooperative} if the $k$th additive compound $Df(x)^{[k]}$ is
Metzler for every $x$; the case $k=1$ is the usual cooperative/Metzler
condition \cite{WeissMargaliot2021, KatzGiordanoMargaliot2025}. Equivalently, the linearised flow $D\varphi_\tau$ is a
\emph{$k$-positive} cocycle: its $k$th multiplicative compound
$\wedge^k D\varphi_\tau$ maps the positive orthant of $\wedge^k
\mathbb R^n$ (in the standard Pl\"ucker basis) into itself. The
additive compound of the Jacobian and the multiplicative compound of
the flow are the infinitesimal and integrated versions of one object:
$\wedge^k D\varphi_\tau$ is exactly the transition matrix of the linear
time-varying system generated by $Df(x(t))^{[k]}$, so ``$Df(x(t))^{[k]}$
Metzler for a.e.\ $t$'' is equivalent to ``$\wedge^k\Phi(t,s)$ positive
for all $t\ge s$'', where $\Phi(t,s)$ is the two-time transition
operator introduced in Section~3 (positivity of the maps
$\wedge^k\Phi(\tau,0)$ from a single fixed initial time alone would not
by itself force the instantaneous generator to be Metzler).

The cocycle is \emph{uniformly $k$-positive} along the trajectory if
this Metzler generator is uniformly so, in the following sense. Write
$B(t) := Df(x(t))^{[k]}$, an $N\times N$ matrix with $N=\binom{n}{k}$.
There are constants $0 < \delta \le \Delta < \infty$ and $\Lambda <
\infty$, independent of $t$, and a \emph{fixed} ($t$-independent)
strongly connected directed graph $G$ on the $N$ coordinates of
$\wedge^k\mathbb R^n$, such that for a.e.\ $t$ along the trajectory:
$B_{ij}(t) \in [\delta,\Delta]$ for every edge $j\to i$ of $G$; every
other off-diagonal entry satisfies $B_{ij}(t) \ge 0$ (Metzler); and all
entries satisfy $|B_{ij}(t)| \le \Lambda$ (the diagonal included). Note
this is a \emph{sparse} hypothesis: the band $[\delta,\Delta]$ is
required only on the edges of $G$, not on every nonnegative off-diagonal
entry, so $G$ need not be complete. These requirements, the lower band
on a fixed connected support, the global bound, and persistent
connectivity, are exactly what the finite-window estimate in the
appendix uses; pointwise irreducibility for a.e.\ $t$ on its own would
not suffice, since the couplings could weaken or the connectivity
pattern change from instant to instant. Imposing the band on the
generator (rather than on $\wedge^k D\varphi_\tau$ itself, whose entries
grow or decay with $\tau$ and so admit no $\tau$-independent band) is
what yields a contraction of the Hilbert projective metric on the
relevant positive cone: the projectivized positive cone, contracting
toward a ray for $k=1$, and the totally positive region
of the Grassmannian $\Gr(k,n)$ (the part represented inside the positive
compound cone) for $k=2$, at a rate $c > 0$ uniform in $\tau$; the
extension to arbitrary transverse planes is the exponential separation
of Section~3, and the dependence of the rate on the band, the global
bound, and the connectivity of the generator is made precise in the next
remark. Composed over $[0,\tau]$, this gives an exponential contraction
$e^{-c\tau}$: the continuous-time limit of the discrete Birkhoff
argument for finite products of uniformly banded positive matrices.

The strong connectivity built into the definition above is what we rely
on twice below (for $k=1$ in Section~1 and $k=2$ in Section~2): at $t=0$
it lets Perron--Frobenius, applied once to the frozen matrix at $x_0$,
deliver a simple dominant eigenvector or eigen-bivector for the
illustrative anchor; along the trajectory it is what makes the Birkhoff
contraction strict, at a rate uniform in $\tau$, as the next remark
records.

\textbf{The Birkhoff rate, precisely.} The claim that the uniform band
produces an exponential rate needs one further ingredient, and it is
the reason irreducibility is imposed along the whole trajectory rather
than only at $x_0$. Over an infinitesimal step the compound flow is
\[
   \exp\!\big(h\,Df(x(t))^{[k]}\big) = I + h\,Df(x(t))^{[k]} + O(h^2),
\]
which sits close to the identity: its off-diagonal entries are $O(h)$
where $Df(x(t))^{[k]}$ is nonzero, and higher order or exactly zero
elsewhere. Such a matrix maps the positive cone into itself but meets
its boundary, so its Hilbert-metric projective diameter is infinite and
it is only non-expansive. A strict contraction appears only once the
transition matrix over a finite window maps the cone into its interior,
that is, becomes entrywise strictly positive.

The fixed strongly connected support, together with the uniform lower
bound, supplies this (pointwise irreducibility for a.e.\ $t$ alone would
not). Suppose $Df(x(t))^{[k]}$ is
Metzler for a.e.\ $t$, with the off-diagonal entries corresponding to
the edges of a fixed strongly connected graph $G$ lying in
$[\delta,\Delta]$, on the $\binom{n}{k} =: N$ coordinates of $\wedge^k
\mathbb R^n$. Let $\Psi(t_2,t_1)$ denote
the transition matrix of the induced compound system $\dot\eta =
Df(x(t))^{[k]}\,\eta$, so that $\wedge^k D\varphi_\tau(x_0) =
\Psi(\tau,0)$. Then there is a window length $h_0 > 0$ for which
$\Psi(t+h_0,t)$ is entrywise strictly positive for every $t$, with
entries bounded above in terms of $\Lambda$, $N$, and $h_0$ and bounded below
in terms of $\delta$, $h_0$, and the graph diameter $r \le N-1$: each
entry receives a contribution of order $(\delta h_0)^r/r!$
from a connecting path of length at most $r$.

For an entrywise positive matrix $P$, Birkhoff's theorem bounds the
contraction coefficient in the Hilbert metric by
\[
   \kappa(P) \le \tanh\!\Big(\tfrac14\,\operatorname{diam}(P)\Big),
   \qquad
   \operatorname{diam}(P) = \max_{i,j,p,q}\log\frac{P_{ip}\,P_{jq}}{P_{iq}\,P_{jp}},
\]
and the entry bounds just stated make the projective diameter of the
window map $\Psi(t+h_0,t)$ finite and uniform in $t$; writing $D_0$ for
its supremum over $t$, the coefficient $\kappa_0 := \tanh(D_0/4) < 1$ is
a uniform contraction bound. Composing over
$\lfloor \tau/h_0\rfloor$ windows,
\[
   d_H\!\big(\wedge^k D\varphi_\tau\,\xi,\ \wedge^k D\varphi_\tau\,\eta\big)
   \;\le\; \kappa_0^{\lfloor \tau/h_0\rfloor}\, d_H(\xi,\eta)
   \;\le\; \kappa_0^{-1}\,e^{-c\tau}\, d_H(\xi,\eta),
\]
which is exponential decay at rate $c = -h_0^{-1}\log\kappa_0
> 0$ (the prefactor $\kappa_0^{-1}$ absorbs the fractional part of
$\tau/h_0$). So the rate depends on the band $[\delta,\Delta]$ and the
global bound $\Lambda$ together with the connectivity: the lower band
$\delta$ and the global bound $\Lambda$ control the sizes of the
entries (the explicit upper estimate in the appendix uses $\Lambda$),
while the fixed connected support, through the graph diameter $r$ and
the window $h_0$, controls how long
they take to become positive. The band alone does not determine it.

\textbf{What this note does and does not prove.} The reasoning below is
deliberately heuristic, so let me be clear about what these claims
actually mean. The alignment statements (``generic tangent
vectors align with a dominant direction'', ``generic 2-planes align
with a dominant 2-plane'') are statements about \emph{directions in the
tangent bundle}. Their clean rigorous form, given in the final section,
is a \emph{finite-time} one: once $T\ge h_0$, the images of interior
cone directions under the compound cocycle have collapsed to within
$e^{-cT}$ of one another (Theorem~\ref{thm:align}), a statement that
rests only on the Birkhoff contraction of the appendix and needs no
compactness, no limiting bundle, and no external input. A canonical
\emph{limiting} object as $T\to\infty$ is an optional refinement,
requiring more (a compact invariant orbit closure), as we will remark.
Birkhoff contraction on the positive cone is
the engine throughout, and it is a statement about in-cone directions,
not, by itself, about arbitrary sign-changing tangent vectors: the
Hilbert metric lives on the cone interior, and passing to a general
tangent vector needs a transversality (no-cancellation) condition,
discussed in the final section. Two further points to keep in mind,
expanded where they arise: the frozen-Jacobian anchors $L(\tau)$ and
$\Pi_{\mathrm{anc}}(\tau)$ introduced below are concrete illustrative
devices, serving as in-cone reference markers, not claimed to be any
canonical invariant bundle; and passing from these tangent-bundle
statements to statements about \emph{finite} nonlinear perturbations at
\emph{large} time is a separate, non-uniform step, taken up in the
final section. A reader who wants only the precise skeleton may read
Theorem~\ref{thm:align} first and treat the intervening sections as
motivation.

\section{Positive cocycle alignment makes the dynamics look one-dimensional}

Here is the picture we are after. A generic tangent vector, transported
by a uniformly positive cocycle, becomes asymptotically aligned with a
single dominant direction. On a non-equilibrium trajectory, the
velocity field $\dot x(\tau)$ is a convenient concrete representative of
that dominant direction (see Alignment below for the precise sense in
which it represents it), and this forces perturbed initial conditions
to end up, to first order and asymptotically in $\tau$, close to the
reference orbit at a slightly shifted time. That is the sense in which
the dynamics collapse to one dimension. The generic long-time
convergence behind this picture is the subject of the theory of
monotone (cooperative) dynamical systems \cite{HirschSmith2005, Smith1995}.

\textbf{Setup.} Consider $\dot x = f(x)$ with smooth $f$, and let
$\varphi_t$ be its flow. Fix a reference trajectory $x(t)=\varphi_t(x_0)$
and let $D\varphi_\tau(x_0)$ denote the derivative of the time-$\tau$ map at
$x_0$, i.e.\ the linearised (variational) propagator. A perturbation
$v_0\in T_{x_0}\mathbb R^n$ evolves as $v(t) = D\varphi_t(x_0)\, v_0$,
which is the solution of the variational equation
$\dot v = Df\bigl(x(t)\bigr)\,v$ with $v(0)=v_0$.

The starting point is a small but useful observation: the velocity
field along the reference trajectory, $\dot x(t) = f(x(t))$, is itself a
solution of the variational equation: differentiating $\dot x = f(x)$ in $t$ gives
$\ddot x = Df(x)\,\dot x$. In cocycle language,
\[
   \dot x(\tau) \;=\; D\varphi_\tau(x_0)\,\dot x(0),
\]
so the velocity is one specific tangent vector transported by the
cocycle.

\textbf{Alignment.} Uniform positivity of the linearised cocycle makes
$D\varphi_\tau(x_0)$ contract the Hilbert projective metric on the
tangent cone. For the asymptotic interpretation in this section, assume
either the compact-base hypotheses of the refinement in Section~3 or,
directly, the existence of the exponentially separated splitting stated
there; under that assumption the cocycle admits
a dominant line bundle $E_1(\tau)$ and a complementary slow bundle
$F_1(\tau)$ of dimension $n-1$, so that every
$v_0 \notin F_1(0)$ has image direction converging to $E_1(\tau)$
exponentially fast. (The unconditional, finite-time version of what
follows, needing no such assumption, is Theorem~\ref{thm:align}.)
Informally, then: images of vectors transverse to the slow
bundle become well-aligned with each other as $\tau$ grows; ``generic''
means avoiding the measure-zero slow bundle.

It is convenient to have a concrete anchor to point at, and we can name
one using only data at $x_0$. Since
$A_0 := Df(x_0)$ is Metzler and irreducible, Perron--Frobenius applied
to $A_0$ (or to $A_0 + cI$, for $c$ large enough to make it entrywise
nonnegative) gives a simple dominant eigenvalue and a strictly
positive eigenvector $p_0$, unique up to scale. Define the transported
line
\[
   L(\tau) \;:=\; D\varphi_\tau(x_0)\,\mathrm{span}(p_0),
\]
exactly, at every $\tau$: we are simply applying the (exactly known)
linear map $D\varphi_\tau(x_0)$ to the fixed line $\mathrm{span}(p_0)$.
($L(\tau)$ itself still drifts with $\tau$, needing not converge to a
fixed limiting line as $\tau\to\infty$; it only keeps pace with whatever
other interior cone direction's image it is compared against.)

Observe that this anchor, as discussed further in Section~3, is
built from frozen data: $p_0$ is the
Perron eigenvector of the \emph{frozen} matrix $A_0$, and the transported
line $L(\tau)$ need not be the dynamical dominant bundle $E_1(\tau)$.
They do not agree exactly but become asymptotically aligned:
$\angle\big(L(\tau), E_1(\tau)\big) \le C e^{-\gamma\tau}$, provided
$p_0 \notin F_1(0)$. Under the Poláčik--Tereščák splitting of Section~3
this transversality is automatic, since $p_0$ is strictly positive and
the slow hyperplane $F_1(0)$ contains no nonzero positive vector; it is
an extra assumption only if one starts instead from an arbitrary,
separately-assumed dominated splitting. Either way,
$L(\tau)$ is used here as a concrete,
computable stand-in for the dominant ray, not as a claim about the
invariant bundle.

With that understood: so long as $\dot x(0) \notin F_1(0)$ (the same
genericity, now applied to the velocity), its image
$\dot x(\tau) = D\varphi_\tau(x_0)\dot x(0)$ becomes exponentially
well-aligned with $E_1(\tau)$, hence with $L(\tau)$ and with the image
of every other generic vector. So $\dot x(\tau)$ is a convenient,
concrete representative of the dominant ray. For any generic $v_0$ there
is then a scalar $\alpha(\tau;v_0)$ (with alignment error decaying
exponentially in $\tau$) such that
\[
   \frac{D\varphi_\tau(x_0)\, v_0}{\lvert D\varphi_\tau(x_0)\, v_0\rvert}
   \;\approx\;
   \pm\frac{\dot x(\tau)}{\lvert\dot x(\tau)\rvert},
   \qquad\text{equivalently}\qquad
   D\varphi_\tau(x_0)\, v_0 \;\approx\; \alpha(\tau;v_0)\,\dot x(\tau),
\]
the alignment being a statement about \emph{directions}; the magnitude
$\alpha(\tau;v_0)$ is a separate matter, taken up in the caveats.

\textbf{The geometric consequence.} Now take a nearby initial state
$y_0 = x_0 + v_0$ and fix the time $\tau$. Two approximations are in
play, and they must be applied in the right order. First, for fixed
$\tau$ and small $\lvert v_0\rvert$, the linearisation
\[
   y(\tau) - x(\tau)
   \;=\; D\varphi_\tau(x_0)\, v_0 + o(\lvert v_0\rvert)
   \;\approx\; \alpha(\tau;v_0)\, \dot x(\tau),
\]
where the second step uses the directional alignment above and
$\alpha(\tau;v_0)$ is the resulting scalar. Second, \emph{provided
$\alpha(\tau;v_0)$ is small}, Taylor-expanding the reference orbit gives
$x(\tau+\alpha) \approx x(\tau) + \alpha\,\dot x(\tau)$, so
\[
   y(\tau) \;\approx\; x\bigl(\tau + \alpha(\tau;v_0)\bigr).
\]
In words: to leading order the perturbed state has landed back on the
reference orbit itself, only at a slightly shifted time $\alpha(\tau;v_0)$. This
is a local, fixed-$\tau$, small-perturbation statement: the order of
limits is $v_0 \to 0$ first (for the linearisation) with $\alpha$
kept small (for the Taylor step), and only then does one ask how the
tangent direction behaves as $\tau\to\infty$. The nonlinear remainder
$o(\lvert v_0\rvert)$ is not uniform in $\tau$, so the two limits do not
commute; Section~3 says what survives.


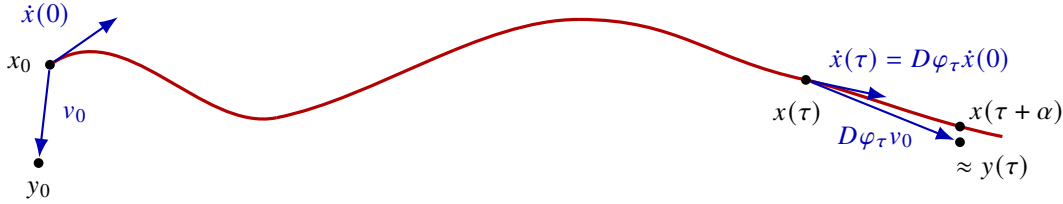
\begin{center}
\begin{tikzpicture}[
    x=1cm, y=1cm,
    every node/.style={font=\small},
    dot/.style={circle, fill=black, inner sep=1.3pt},
    traj/.style={very thick, red!70!black},
    tangent/.style={-{Latex[length=2.5mm]}, thick, blue!70!black},
    pert/.style={-{Latex[length=2.5mm]}, thick, blue!70!black}
  ]

  \draw[traj] (0,0) .. controls (0.983,0.688) and (2.0,-0.9) .. (3.0,-0.7);
  \draw[traj] (3.0,-0.7) .. controls (4.3,-0.44) and (5.7,0.6) .. (7.0,0.6);
  \draw[traj] (7.0,0.6) .. controls (8.3,0.6) and (8.729,0.070) .. (10.0,-0.2);
  \draw[traj] (10.0,-0.2) .. controls (10.587,-0.325) and (11.6,-0.75) .. (12.6,-0.95);

  \node[dot] (x0) at (0,0) {};
  \node[left=3pt] at (x0) {$x_0$};
  \draw[tangent] (x0) -- (0.901,0.631);
  \node[above left=2pt] at (0.4505,0.3155) {\textcolor{blue!70!black}{$\dot x(0)$}};

  \node[dot] (y0) at (-0.15,-1.3) {};
  \node[below=3pt] at (y0) {$y_0$};
  \draw[pert] (x0) -- (y0);
  \node[right=3pt] at (-0.075,-0.65) {\textcolor{blue!70!black}{$v_0$}};

  \node[dot] (xT) at (10.0,-0.2) {};
  \node[below left=4pt and -10pt] at (xT) {$x(\tau)$};
  \draw[tangent] (xT) -- (11.076,-0.429);
  \node[above right=2pt and -10pt] at (10.538,-0.3145)%
       {\textcolor{blue!70!black}{$\dot x(\tau)=D\varphi_\tau\dot x(0)$}};

  \node[dot] (yT) at (12.039,-1.024) {};
  \node[below left=1pt and -30pt] at (yT) {$\approx y(\tau)$};
  \draw[pert] (xT) -- (yT);
  \node[below=3pt] at (10.918,-0.571) {\textcolor{blue!70!black}{$D\varphi_\tau v_0$}};

  \node[dot] (xTalpha) at (12.039,-0.817) {};
  \node[above right=-3pt and 0pt] at (xTalpha) {$x(\tau+\alpha)$};
\end{tikzpicture}
\end{center}

\textbf{Remarks.}
%
Two clarifications:

\emph{Equilibria.} At a fixed point $f(x_*)=0$ so $\dot x=0$; there is no
distinguished direction $\dot x(\tau)$ for perturbations to align with, and
the collapse to the reference orbit becomes vacuous. The intuition
applies away from such points.

\emph{Magnitude versus direction.} The alignment is projective: it
controls the \emph{direction} of $D\varphi_\tau(x_0)v_0$, not its
length. If the scalar $\alpha(\tau;v_0)$ happens to be small, the
landing statement $y(\tau)\approx x(\tau+\alpha)$ above holds. If
instead $\alpha(\tau;v_0)$ is large, that statement fails: the tangent
line at $x(\tau)$ leaves a curved orbit immediately, so
$x(\tau)+\alpha\,\dot x(\tau)$ is \emph{not} close to
$x(\tau+\alpha)$, and one cannot conclude that the perturbed state
sits near a distant point of the reference orbit. What remains true
without smallness is only the directional (projective) statement: the
perturbation's image points along $\dot x(\tau)$. The clean, unconditional
content is therefore about tangent directions; the landing-on-the-orbit
picture is its small-shift, first-order shadow.

So: in the projective / first-order sense, and provided the velocity is
itself transverse to the slow bundle ($\dot x(0)\notin F_1(0)$, so that
$\dot x(\tau)$ represents the dominant direction), every normalised
tangent perturbation $v_0\notin F_1(0)$ becomes aligned with
$\dot x(\tau)$.
That is a reasonable one-line reading of ``uniformly positive cocycle
$\Rightarrow$ dynamics look one-dimensional'', with the understanding
that it is a statement about tangent directions, upgraded to a landing
on the reference orbit only for small shifts.

\section{2-positive cocycle alignment makes the dynamics look two-dimensional}

In the 1-positive case, uniform positivity of the linearised cocycle
forced generic tangent perturbations to become asymptotically aligned
with a single direction, represented concretely by the flow
$\dot x(\tau)$; in the fixed-time, small-shift, first-order picture,
perturbed states therefore landed close to the 1-D reference orbit
itself.

In the 2-positive case, the analogue turns out to be a statement about \emph{pairs}
of tangent vectors, or equivalently about the 2-planes they span.
Uniform 2-positivity gives Birkhoff contraction of the wedge cocycle on
the positive cone of $\wed\mathbb R^n$, not on individual directions.
Directly, this contracts the positive-Pl\"ucker region (the totally
positive part of $\Gr(2,n)$): the images of any two such planes become
exponentially close to one another, so at each time the family may be
represented by a concrete in-cone marker, for instance the transported
frozen plane $\Pi_{\mathrm{anc}}(\tau)$ constructed below. What direct
contraction does \emph{not} by itself supply is a canonical invariant
object: only under the compact-base exponential-separation refinement
of Section~3 is the contracting family represented canonically by the
invariant dominant plane $\Pi_*(\tau) = E_2(\tau) \subset
T_{x(\tau)}\mathbb R^n$. Alignment of an \emph{arbitrary} transverse
plane (not just positive-Pl\"ucker ones) likewise belongs to that
refinement. Granting it, single perturbations no longer collapse
to a line: since $\varphi_\tau$ is locally a diffeomorphism, the
nonlinear image of a small $n$-dimensional ball remains an
$n$-dimensional neighbourhood (a distorted one, not literally an
ellipsoid). To first order it is represented by the ellipsoid
$D\varphi_\tau(x_0)(\text{ball})$, which becomes increasingly thin in
the $n-2$ relative directions transverse to $\Pi_*(\tau)$, so the
dominant tangent geometry is two-dimensional and well approximated by
$\Pi_*(\tau)$. The dominant element is always a
genuine 2-plane, in every dimension $n$: this is a feature of the
exterior-power structure, not a low-dimensional accident, and we return
to it below. The figure is drawn for $n=3$.

\textbf{Setup on the wedge power.} The variational cocycle induces a
cocycle $\wed D\varphi_\tau(x_0)$ on the second exterior power
$\wed T_{x_0}\mathbb R^n$. On decomposable 2-vectors,
\[
   \wed D\varphi_\tau\,(v_0 \wedge w_0)
   \;=\;
   D\varphi_\tau v_0 \,\wedge\, D\varphi_\tau w_0,
\]
i.e., it transports 2-planes: if $\Pi_0 = \mathrm{span}(v_0, w_0)$ is a
2-plane in $T_{x_0}\mathbb R^n$, then its image under the cocycle is the
2-plane $\mathrm{span}(D\varphi_\tau v_0, D\varphi_\tau w_0)$ in
$T_{x(\tau)}\mathbb R^n$. Uniform 2-positivity makes $\wed D\varphi_\tau$
a Birkhoff contraction of the Hilbert projective metric on the positive
orthant of $\wed \mathbb R^n$. Two things need care in how this transfers
to the Grassmannian. First, a decomposable 2-vector $v_0\wedge w_0$ has
all its Pl\"ucker coordinates of one sign only for planes in a
particular region; a generic oriented 2-plane does not sit in the
positive orthant, so cone contraction does not act directly on all of
$\Gr(2,n)$. What is true is that planes whose Pl\"ucker representative
lies in the interior of the compound cone are contracted toward one
another, and the clean statement for arbitrary transverse planes is
again the exponential separation of Section~3 (of which this cone
contraction is the mechanism). Second, because the operator is the
compound $\wed D\varphi_\tau$ and not an arbitrary positive map, it
carries decomposables to decomposables exactly; this is what keeps
the dominant object a genuine 2-plane in every dimension, a point we
return to below.

\textbf{Alignment to a 2-plane.} Granting the index-2 exponential
separation of Section~3, which supplies a dominant 2-plane bundle
$E_2(\tau) = \Pi_*(\tau)$ and a complementary slow bundle $F_2(\tau)$ of
dimension $n-2$, any 2-plane $\Pi_0$ with $\Pi_0 \cap F_2(0) = \{0\}$ is
transported exponentially close to $\Pi_*(\tau)$. The exceptional set is
the 2-plane analogue of the ``bad directions'' in Section~1: a starting
2-plane that meets the slow bundle $F_2(0)$ nontrivially fails to reach
the dominant 2-plane. Concretely it is the incidence locus
\[
   \{\Pi_0 \in \Gr(2,n) : \Pi_0 \cap F_2(0) \ne \{0\}\},
\]
a special Schubert variety: choosing any complete flag with
$F^{n-2} = F_2(0)$, it is the Schubert variety of 2-planes meeting
$F^{n-2}$ nontrivially, of codimension one in $\Gr(2,n)$, hence
lower-dimensional and of measure zero, so a
2-plane chosen at random avoids it with probability one. This is what
``generic'' means here. The one thing to flag is that the slow bundle
$F_2(0)$ is not produced by cone contraction for free: its existence
\emph{is} the exponential-separation statement of Section~3, and the
Schubert description is meaningful only once that splitting is in hand.

\textbf{Identifying the dominant 2-plane.} As in Section~1, there is an
exact, elementary way to name a concrete anchor here too, again using
only data at $x_0$. Since $Df(x_0)^{[2]}$ is Metzler and irreducible,
Perron--Frobenius applied to this compound matrix, exactly as before,
gives a simple dominant eigen-bivector $\xi_0 \in \wed T_{x_0}\mathbb
R^n$, unique up to scale.

It turns out that $\xi_0$ is decomposable, $\xi_0 = a_0\wedge b_0$, in
\emph{every} dimension $n$, and so always determines a genuine 2-plane
$V_2(x_0) := \mathrm{span}(a_0,b_0)$. This is not obvious from counting
dimensions, since decomposable bivectors are a proper subvariety of
$\wed\mathbb R^n$ for $n\ge4$; it is a consequence of the exterior-power
structure. The point is that $e^{t\,Df(x_0)^{[2]}} = \wed e^{t\,Df(x_0)}$
is itself a second compound, so it carries decomposable bivectors to
decomposable bivectors for every $t$. Starting from any interior
positive decomposable bivector (for instance $u\wedge v$ with
$u=(1,\dots,1)$ and $v=(t_1,\dots,t_n)$ strictly increasing, whose
Pl\"ucker coordinates $t_j-t_i>0$ are all positive), the normalised
iterates $e^{t\,Df(x_0)^{[2]}}(u\wedge v)/\lVert\cdot\rVert$, where
$\lVert\cdot\rVert$ is the Euclidean norm on $\wed\mathbb R^n$ in the
standard Pl\"ucker basis, stay
decomposable and converge to the Perron ray of $Df(x_0)^{[2]}$; since
the decomposable bivectors form a closed set in
$\mathbb P(\wed\mathbb R^n)$ (they are the image of the Grassmannian
under Pl\"ucker), the limiting Perron ray is decomposable too. So the
frozen anchor is a genuine plane regardless of $n$; the special role of
$n=3$ elsewhere is only that the figure is drawn there. (For general
$k$ the same argument applies, needing an interior positive
decomposable $k$-vector to start from; choosing points $0<t_1<\cdots<t_k$
and taking the span of the $k$ vectors $(1,t_a,t_a^2,\dots,t_a^{n-1})$,
$a=1,\dots,k$, provides one, since every ordered $k\times k$ minor is a
generalized Vandermonde determinant and is strictly positive.)

There is also a direct downstairs route to the same frozen plane,
provided by the generalized Perron theorem of Fusco and Oliva
\cite{FuscoOliva1991}. Put $A_0 = Df(x_0)$ and $T_s = e^{sA_0}$, and let
\[
   P_-^k = \{z\in\mathbb R^n : s^-(z)\le k-1\}
\]
be the closed sign-variation cone of rank $k$, where $s^-(z)$ is the
number of sign changes in the coordinate sequence of $z$ after deleting
zeros. This is a generally nonconvex cone that contains $k$-dimensional
subspaces but no subspace of dimension $k+1$, so it is a ``cone of
dimension $k$'' in their sense (and lives \emph{downstairs}, in
$\mathbb R^n$, not in $\wedge^k\mathbb R^n$). Since $A_0^{[k]}$ is
irreducible Metzler, the $k$th \emph{multiplicative} compound of $T_s$
satisfies
\[
   (T_s)^{(k)} = \big(e^{sA_0}\big)^{(k)} = e^{s A_0^{[k]}} \gg 0
   \qquad (s>0),
\]
(the multiplicative compound of the exponential is the exponential of
the additive compound; we write $(\cdot)^{(k)}$ for the former and
$(\cdot)^{[k]}$ for the latter),
so $T_s$ is strictly sign-regular of order $k$ and the
variation-diminishing property \cite{MairhuberSchoenbergWilliamson1959, AlseidiMargaliotGarloff2019}
gives $T_s(P_-^k\setminus\{0\})\subset
\operatorname{int}P_-^k$. Fusco--Oliva Theorem~1 (with $E=\mathbb R^n$,
$K=P_-^k$, $d=k$, $T=T_s$) then yields unique complementary
$T_s$-invariant subspaces $W_1, W_2$, with $\dim W_1 = k$, $\dim W_2 =
n-k$, $W_1\setminus\{0\}\subset\operatorname{int}P_-^k$, $W_2\cap P_-^k
= \{0\}$, and a strict spectral-modulus gap between $T_s|_{W_1}$ and
$T_s|_{W_2}$. Because $e^{rA_0}$ commutes with $T_s$ and maps
$P_-^k\setminus\{0\}$ into its interior, $e^{rA_0}W_1$ is again a
$k$-dimensional $T_s$-invariant subspace with nonzero vectors in
$\operatorname{int}P_-^k$; by uniqueness $e^{rA_0}W_1 = W_1$ for every
$r>0$, and differentiating at $r=0$ gives $A_0 W_1\subset W_1$. So
$W_1 = V_2(x_0)$ (for $k=2$) is the frozen $A_0$-invariant plane we
want.

The two routes agree. Because $W_1$ is $k$-dimensional and
$T_s$-invariant, $\wedge^k W_1$ is a one-dimensional invariant subspace
of $(T_s)^{(k)} = e^{sA_0^{[k]}}$. By the standard Pl\"ucker-coordinate
characterization of $k$-planes contained in $\operatorname{int}P_-^k$,
the orientation of $W_1$ can be chosen so that all of its ordered
Pl\"ucker coordinates are strictly positive; then $\wedge^k W_1$ is a
positive invariant ray, and since the positive compound operator has a
unique positive Perron ray, $\wedge^k W_1 = \operatorname{span}\{\xi_0\}$,
the compound Perron vector above. Thus Fusco--Oliva and the compound
argument identify the same frozen $k$-plane: the former works directly
with a rank-$k$ (and generally nonconvex) cone in $\mathbb R^n$, the
latter with the ordinary
positive cone in $\wedge^k\mathbb R^n$. In their \S3, Fusco and Oliva
apply the theorem to positive Jacobi (tridiagonal) matrices, recovering
the Gantmacher--Krein oscillation theory \cite{GantmacherKrein1937};
that is a special illustration of their result, not the general
situation here, since the frozen Jacobian $Df(x_0)$ of a $k$-positive
system is an arbitrary matrix with $Df(x_0)^{[k]}$ irreducible Metzler
and need not have the tridiagonal (nearest-neighbour) structure of a
Jacobi matrix. We keep the elementary
compound-Perron argument above for self-containedness;
Fusco--Oliva gives the direct rank-$k$-cone formulation for a single
frozen operator, and Pol\'a\v{c}ik--Tere\v{s}\v{c}\'ak, used in
Section~3, supplies the corresponding cocycle separation after lifting
to the ordinary positive cone in $\wedge^k\mathbb R^n$, giving the
$t$-varying bundle $E_k(\tau)$ rather than the frozen anchor.

Transporting this frozen plane, define
\[
   \Pi_{\mathrm{anc}}(\tau) \;:=\; D\varphi_\tau(x_0)\, V_2(x_0),
\]
at every $\tau$. As with the line $L(\tau)$ in Section~1, this is an
illustrative anchor, not the dynamical dominant bundle: the Perron data
of the frozen compound $Df(x_0)^{[2]}$ has no a priori reason to
coincide with $\Pi_*(\tau) = E_2(\tau)$. What holds, provided the
anchor plane is transverse to $F_2(0)$, is exponential closeness,
\[
   \operatorname{dist}_{\Gr}\!\big(\Pi_{\mathrm{anc}}(\tau),\,\Pi_*(\tau)\big)
   \;\le\; C e^{-\gamma\tau},
\]
so $\Pi_{\mathrm{anc}}(\tau)$ approaches, but does not equal,
$\Pi_*(\tau)$. Under the strong-positivity mechanism used here the
frozen plane has a positive Pl\"ucker representative, which makes this
transversality automatic.

Provided $\mathrm{span}\{\dot x(0), w(0)\} \cap F_2(0) = \{0\}$ (a
\emph{different} transversality condition from the one just above: that
one concerned the frozen anchor plane $V_2(x_0)$, this one concerns the
plane spanned by the velocity and the chosen companion, and note it is a
condition on the \emph{span}, not merely on each of $\dot x(0)$ and
$w(0)$ separately, since two vectors can each avoid $F_2(0)$ while their
span still meets it), the
transported plane $\mathrm{span}(\dot x(\tau),\, D\varphi_\tau w(0))$
becomes exponentially well-aligned with $\Pi_*(\tau)$; in this
asymptotic sense both $\dot x(\tau)$ and $D\varphi_\tau w(0)$ come to
lie in $\Pi_*(\tau)$, neither exactly at finite $\tau$. Define $w(\tau)
:= D\varphi_\tau(x_0)\, w(0)$, the actual transported vector (not merely
its direction), so that later $\partial_s X(\tau;0) = w(\tau)$ makes the
family-parameter shift unambiguous. A caution on naming:
$w(\tau)$ is merely \emph{a} second vector spanning the transported
plane together with $\dot x(\tau)$, chosen through the companion
$w(0)$; it is not a canonical ``second-slowest mode''. A dominant
2-plane need not split into a first and a second direction at all, the
dynamics inside it may rotate or shear and carry no invariant line
field, so there is in general no dynamically distinguished second axis.
Together $\dot x(\tau)$ and $w(\tau)$ span a 2-plane approximating
$\Pi_*(\tau)$ ever better as $\tau$ grows; the figure below shows this
schematically, at some implicitly large but finite $\tau$.

\textbf{Consequence for a single perturbation.} Take a further
perturbation $v(0)\in T_{x_0}\mathbb R^n$ (independent of $\dot x(0)$
and $w(0)$, and with $v(0)\notin F_2(0)$, so that its image is not
trapped in the slow bundle), and let $y_0 = x_0 + v(0)$. In the
1-positive setting, $D\varphi_\tau v(0)$ became asymptotically aligned
with a specific line; under index-2 separation alone, without an
additional index-1 separation, this need not happen,
but $D\varphi_\tau v(0)$ becomes asymptotically aligned with
$\Pi_*(\tau)$. To first order in $|v(0)|$,
\[
   y(\tau) - x(\tau)
   \;\approx\;
   D\varphi_\tau v(0)
   \;\approx\;
   \alpha(\tau;v(0))\, \dot x(\tau) + \beta(\tau;v(0))\, w(\tau),
\]
where the coefficients $\alpha(\tau;v(0)), \beta(\tau;v(0))$ depend on
both $v(0)$ and the time $\tau$. So $y(\tau)$ lies, to first order
and asymptotically in $\tau$, in the 2-plane through $x(\tau)$ spanned by
$\dot x(\tau)$ and $w(\tau)$; the vector labelled $v(\tau)$ in the figure marks
a generic direction of this kind inside the transported plane (see the
remark below the figure on how it is drawn).

Different initial perturbations $v(0), w(0), \ldots$ (each with
$v(0)\notin F_2(0)$, so that its image is not trapped in the slow
bundle) give different coefficient pairs. Since each $\varphi_\tau$ is
a local diffeomorphism, the nonlinear image of a small $n$-dimensional
ball remains an $n$-dimensional neighbourhood, generally distorted, not
literally an ellipsoid and not a 2-dimensional set; what is
two-dimensional is the \emph{dominant} geometry. To first order the
image is represented by the ellipsoid $D\varphi_\tau(x_0)(\text{ball})$,
which becomes increasingly thin in its $n-2$ relative
directions transverse to $\Pi_*(\tau)$, so the dominant tangent
geometry is well approximated by $\Pi_*(\tau)$. To get a genuine
surface one must \emph{select} a transverse family: a single
one-parameter family of nearby initial conditions, evolved in time,
sweeps out a genuine immersed 2-surface $\Sigma$ containing the
reference orbit (provided its two tangent vectors stay independent),
and $\Pi_*(\tau)$ is asymptotically (not exactly) its tangent plane at
$x(\tau)$. Arbitrary nearby perturbations do \emph{not} lie on that
selected surface; different transverse families give different surfaces
with different tangent planes along the orbit, and there is no single
canonical envelope through which all nearby trajectories pass. With that
understood, $\Pi_*(\tau)$ plays for 2-positive systems the role the
tangent line played for 1-positive systems in Section~1, and a chosen
$\Sigma$ plays the role of the reference orbit.

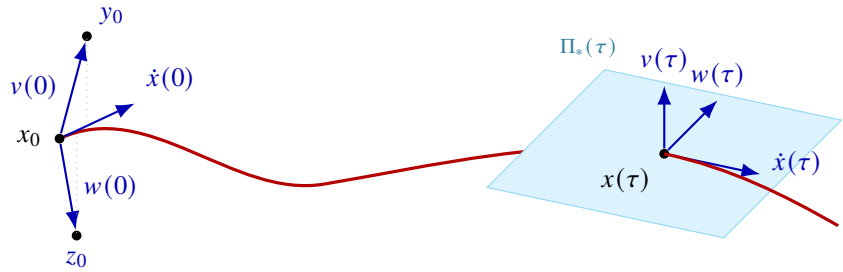
\begin{center}
\begin{tikzpicture}[
    x=1cm, y=1cm,
    every node/.style={font=\small},
    dot/.style={circle, fill=black, inner sep=1.3pt},
    traj/.style={very thick, red!70!black},
    persp/.style={dashed, thick},
    tangent/.style={-{Latex[length=2.5mm]}, thick, blue!70!black},
    leader/.style={gray!55, dotted, thin}
  ]


  \draw[traj] (0,0) .. controls (1.178,0.549) and (2.3,-0.8) .. (3.5,-0.6);
  \draw[traj] (3.5,-0.6) .. controls (4.802,-0.383) and (6.728,0.070) .. (8.0,-0.2);

  \fill[cyan!12] (5.657,-0.645) -- (8.787,-1.310) -- (10.343,0.245) -- (7.213,0.910) -- cycle;
  \draw[cyan!40, thin] (5.657,-0.645) -- (8.787,-1.310) -- (10.343,0.245) -- (7.213,0.910) -- cycle;
  \node[cyan!55!black, font=\scriptsize\itshape] at (7.0,1.2) {$\Pi_*(\tau)$};

  \node[dot] (xT) at (8.0,-0.2) {};
  \node[below left=2pt] at (xT) {$x(\tau)$};
  \draw[tangent] (xT) -- (9.272,-0.470);
  \node[right=1pt, yshift=4pt] at (9.272,-0.470) {\textcolor{blue!70!black}{$\dot x(\tau)$}};
  \draw[tangent] (xT) -- (8.707,0.507);
  \node[above=1pt] at (8.707,0.507) {\textcolor{blue!70!black}{$w(\tau)$}};
  \draw[tangent] (xT) -- (8.0,0.7);
  \node[above=1pt] at (8.0,0.7) {\textcolor{blue!70!black}{$v(\tau)$}};

  \draw[traj] (8.0,-0.2) .. controls (8.880,-0.387) and (9.6,-0.75) .. (10.3,-1.15);

  \node[dot] (x0) at (0,0) {};
  \node[left=3pt] at (x0) {$x_0$};
  \draw[tangent] (x0) -- (0.997,0.465);
  \node[above right=1pt] at (0.997,0.465) {\textcolor{blue!70!black}{$\dot x(0)$}};

  \node[dot] (y0) at (0.362,1.352) {};
  \node[above right=1pt] at (y0) {\textcolor{blue!70!black}{$y_0$}};
  \draw[tangent] (x0) -- (y0);
  \node[left=2pt] at (0.181,0.676) {\textcolor{blue!70!black}{$v(0)$}};
  \draw[leader] (0.362,1.32) -- (0.362,0.05);

  \node[dot] (z0) at (0.226,-1.280) {};
  \node[below=2pt] at (z0) {\textcolor{blue!70!black}{$z_0$}};
  \draw[tangent] (x0) -- (z0);
  \node[right=2pt] at (0.113,-0.640) {\textcolor{blue!70!black}{$w(0)$}};
  \draw[leader] (0.226,-1.25) -- (0.226,0.05);

\end{tikzpicture}
\end{center}

\emph{The plane $\Pi_*(\tau)$ is drawn in local coordinates centred at
$x(\tau)$. The curve passes through the base point $x(\tau)$, with
tangent $\dot x(\tau)$ exponentially close to $\Pi_*(\tau)$, but it is
not generally contained in that fixed affine plane, being curved. As
$\tau$ varies, $\Pi_*(\tau)$ varies continuously, remaining
(asymptotically) close to the tangent plane of the fixed, undrawn
curved surface $\Sigma$ at the point $x(\tau)$. The arrows $v(\tau)$
and $w(\tau)$ are schematic vectors drawn inside $\Pi_*(\tau)$ for
legibility; at finite $\tau$ they need not belong to it exactly.}

\textbf{On the placement of $v(\tau)$ and $w(\tau)$.} The vectors $v(\tau)$ and
$w(\tau)$ drawn inside $\Pi_*(\tau)$ are positioned by a drawing convention,
a legible choice of direction within the plane, not by computing
$D\varphi_\tau v(0)$ and $D\varphi_\tau w(0)$ explicitly; no pushforward was
carried out. The qualitative content of the figure is nonetheless
correct: since $\Pi_*(\tau)$ is (asymptotically, as $\tau$ grows) exactly the
set of directions that transported vectors align with, $D\varphi_\tau
v(0)$ and $D\varphi_\tau w(0)$ genuinely do become well-aligned with
$\Pi_*(\tau)$ to first order. What the figure does \emph{not} assert is
that $v(\tau)$ is the specific image of $v(0)$, or $w(\tau)$ the specific
image of $w(0)$. It illustrates only the general phenomenon, namely that
directions spread apart at time $0$ become approximately coplanar at
time $\tau$, rather than a particular computed trajectory for either
vector.

\textbf{Interpretation of the two coordinates.} Concretely, let
\[
   X(t;s) \;:=\; \varphi_t\bigl(x_0 + s\, w(0)\bigr),
\]
defined for $t\ge0$ and $s$ near $s_0 = 0$. Then $X(t;0) =
x(t)$ recovers the reference orbit, and $\partial_s X(\tau;0) =
D\varphi_\tau(x_0) w(0) = w(\tau)$ exactly, with $w(\tau)$ the
transported vector defined above: varying $t$ near $\tau$ moves along
the reference orbit itself (direction $\dot x(\tau)$), while varying $s$
near $s_0$ moves transversally (direction $w(\tau)$); together
they trace out $\Sigma = \{X(t;s)\}$. This surface depends on the chosen
transverse direction $w(0)$: a different transverse family produces a
different surface, with a different tangent plane along the reference
orbit. There is no ``the'' envelope independent of that choice; $\Sigma$
is one representative surface, adapted to $w(0)$.

The $\alpha(\tau;v(0))\,\dot x(\tau)$
component shifts $y(\tau)$ along the reference orbit, i.e.,
$x(\tau) + \alpha\,\dot x(\tau) \approx x(\tau + \alpha)$ for small
$\alpha$ (the same small-shift proviso as in Section~1). The
$\beta(\tau;v(0))\,w(\tau)$ component is transverse to the orbit inside
$\Sigma$. Because $\partial_s X(\tau;0) = w(\tau)$ exactly, the same
coefficient $\beta$ serves as the first-order shift in the family
parameter $s$: within the chosen family $X(t;s)$, the transverse
coordinate simply moves to a neighbouring member, so the combined small
shift reads
\[
   y(\tau) \;\approx\; X\big(\tau + \alpha;\ s_0 + \beta\big),
\]
i.e.\ the perturbed state lands, to first order, on a nearby orbit of
this family at a slightly shifted time. (No claim is made that the
dominant 2-plane distribution integrates to a foliation: that would
require a Frobenius involutivity condition which positivity does not
supply. The statement here is only about the one chosen family $\Sigma$,
where a foliation by the $s$-parameter orbits holds by construction.
This is exactly the point at which the differential-positivity theory
of Mostajeran and Sepulchre \cite{MostajeranSepulchre2018} imposes
\emph{involutivity} of the rank-$k$ dominant distribution in order to
obtain a genuine $k$-dimensional integral submanifold attractor; the
finite-time alignment developed here makes no such hypothesis, at the
cost of speaking about tangent directions rather than an invariant
submanifold.)

\textbf{Remarks.} This is a first-order picture on the tangent bundle,
pushed down to state space; it is not a statement about exact nonlinear
orbits, and (as in Section~1) the passage to finite perturbations at
large time is the non-uniform step discussed in Section~3. Equilibria
are exceptional. And if the cocycle is \emph{also} uniformly 1-positive,
so that it has a dominant line bundle $E_1(\tau)$ inside $E_2(\tau)$,
then within the envelope a further alignment to $E_1$ collapses the
picture back to the 1-D one, \emph{provided} the velocity is
nondegenerate for that finer splitting, i.e.\ $f(x_0) \notin F_1(x_0)$
(otherwise $\dot x$ lies in the slower complement and does not represent
$E_1$). So 2-positivity yields a dominant two-dimensional tangent
geometry; a chosen transverse family then supplies one illustrative
envelope surface, with no canonical envelope singled out. Whether the
dynamics inside the dominant plane collapse further is a separate
(1-positive) question, subject to that nondegeneracy.

\section{Making the alignment precise}
\label{sec:rigorous}

The two sections above are heuristic. It turns out that the skeleton
behind them can be stated precisely, and this section does so. The
primary statement is a \emph{finite-time} one: for every $T\ge0$ the
compound map satisfies a window-by-window pairwise contraction estimate,
and once $T\ge h_0$ the image of the entire nonzero cone has finite
projective diameter, decreasing exponentially with $T$ (so interior
cone directions have collapsed to within $e^{-cT}$ of one another). This needs nothing
beyond the appendix estimate, without assuming compactness or the
existence of invariant bundles. Only afterwards, as an optional refinement, do we ask
what happens as $T\to\infty$ (a canonical limiting bundle, requiring
more), and finally what is lost in passing from tangent directions to
finite nonlinear perturbations. Statements are given with sketches
rather than full proofs; the aim is to isolate exactly what must be
assumed and what then follows.

Throughout, $\Phi(t,s) := D\varphi_t(x_0)\,D\varphi_s(x_0)^{-1}$ is the
variational transition operator along the reference trajectory, $d_H$ is
the Hilbert projective metric on the positive cone $K_k$ of
$\wedge^k\mathbb R^n$, and $\angle$ denotes the (projective) angle
between directions or subspaces.

\begin{theorem}[Finite-time mutual alignment]
\label{thm:align}
Assume uniform $k$-positivity along the trajectory on $[0,T]$ (the
standing hypotheses, quantified in the appendix). Write $P_T :=
\wedge^k\Phi(T,0)$ for the $k$th compound of the transition operator,
let $K_k$ be the positive cone of $\wedge^k\mathbb R^n$, and set
$m := \lfloor T/h_0\rfloor$. There are constants $D_0<\infty$ and
$\kappa_0\in(0,1)$, and a rate $c = -h_0^{-1}\log\kappa_0 > 0$,
depending only on the positivity data $(\delta,\Lambda,N,r,h_0)$ and not
on $T$, such that:
\begin{enumerate}
\item[(i)] \emph{(Pairwise contraction.)} For all $\xi,\eta$ in the
interior of $K_k$,
\[
   d_H\!\big(P_T\,\xi,\ P_T\,\eta\big)
   \;\le\; \kappa_0^{\,m}\, d_H(\xi,\eta).
\]
\item[(ii)] \emph{(Shrinking image.)} For $T\ge h_0$ the image cone has
finite, exponentially small projective diameter,
\[
   \operatorname{diam}_H\big(P_T(K_k\setminus\{0\})\big)
   \;\le\; D_0\,\kappa_0^{\,m-1} \;=\; O(e^{-cT}).
\]
(For every $T>0$ the image already lies in the open cone, so its
projective diameter is finite; what the bound above needs is $T\ge h_0$,
enough for at least one complete window, so that the diameter is not
merely finite but exponentially small. As $T\to0$ the diameter diverges
and the contraction coefficient tends to $1$, which is why a fixed
window is required, as explained after the sketch.)
\end{enumerate}
Consequently, for $T\ge h_0$, all interior cone directions are carried
to within $O(e^{-cT})$ of one another. Fixing any strictly positive
in-cone reference makes this concrete: the transported Perron object
$L(T) = \Phi(T,0)\,\mathrm{span}(p_0)$ for $k=1$, or
$\Pi_{\mathrm{anc}}(T) = \Phi(T,0)\,V_2(x_0)$ for $k=2$, is in the cone
(since $p_0$, and the Perron bivector of $Df(x_0)^{[2]}$, are strictly
positive), and every interior cone direction has image within
$O(e^{-cT})$ of it. The velocity direction $\dot x(T)$ is not
automatically an in-cone reference: $\dot x(0)$ is a generic tangent
vector, not \emph{a priori} in the cone. It is a valid reference at
horizon $T$ precisely when $P_T\dot x(0)$ does not suffer the
cancellation described in the next remark; under the asymptotic
refinement of Section~3 this becomes the cleaner condition
$\dot x(0)\notin F_1(0)$.
\end{theorem}

\emph{Sketch.} For (i), each window of length $h_0$ is a strict Birkhoff
contraction with coefficient $\kappa_0<1$ uniform in $t$ (appendix), and
$P_T$ is the product of $m$ such windows (the final shorter remainder is
non-expansive on the cone), giving the stated bound. Claim (ii) is a
\emph{separate} argument, since $K_k$ has infinite Hilbert diameter and
one cannot simply multiply that by $\kappa_0^m$: instead, the first
complete window already maps $K_k$ into a set of projective diameter at
most $D_0$ (the entrywise-positive band of the window map, appendix),
and each of the remaining $m-1$ windows contracts \emph{that finite
diameter} by $\kappa_0$, yielding $D_0\kappa_0^{m-1}$. The reference
statement is then immediate, since a strictly positive reference lies in
the image cone. \hfill$\square$

It is worth being precise about the role of the window length, since it
is easy to misstate. Under the fixed strongly connected support assumed
here, $\Psi(t+h,t) = I + hB(t) + o(h)$ is entrywise strictly positive
for \emph{every} $h>0$: entries joining coordinates at graph distance
$q$ enter at order $h^q$ (this is the appendix computation). So every
positive-duration map is already a strict Birkhoff contraction, and the
image cone has finite projective diameter for every $T>0$. What fails as
$h\to0$ is not strict positivity but \emph{uniformity}: the long-path
entries are $O(h^q)$, so the projective diameter of the window image
diverges and its Birkhoff coefficient tends to $1$. A fixed window
$h_0>0$ is needed not to obtain strict positivity, but to obtain a
contraction coefficient $\kappa_0$ bounded uniformly away from $1$,
which is what makes the composite bound decay in $T$.

Two features of this statement deserve emphasis, because they are what
make it self-contained. First, the reference direction is a matter of
convenience among \emph{in-cone} markers: the theorem says all interior
cone directions are mutually aligned, so any strictly positive marker
will do, and the transported Perron object $L(T)$ (or
$\Pi_{\mathrm{anc}}(T)$ for $k=2$) is simply the most concrete one
available. The velocity is not automatically such a marker, since
$\dot x(0)$ need not lie in the cone; it earns the role only under the
condition discussed in the next remark. Second, the ``sliver''
$P_T K_k$ drifts with $T$: the theorem does not claim it
converges to a fixed limiting object. This drift is compatible with
mutual alignment being exponentially tight at each fixed $T$, since
alignment constrains only the relative positions of images within the
sliver, not the location of the sliver itself.

\textbf{General directions and the role of transversality.} The
cone statement above concerns directions in the positive cone $K_k$. A
general element $\xi_0\in\wedge^k\mathbb R^n$ (for $k=1$ a tangent
vector; for $k\ge2$ a decomposable $\xi_0$ representing a $k$-plane)
need not have a positive representative, and here a transversality
condition genuinely enters, even at finite $T$. Split $\xi_0 = \xi_0^+
- \xi_0^-$ into coordinatewise positive and negative parts, both in the
cone. (For $k\ge2$ the parts $\xi_0^\pm$ need not themselves be
decomposable; that is harmless here, as the estimate below is purely
about cancellation in $\wedge^k\mathbb R^n$.) Then
\[
   P_T\,\xi_0
   = P_T\,\xi_0^+ \;-\; P_T\,\xi_0^- ,
\]
and both terms land in the same thin sliver, hence are nearly parallel.
But $P_T\xi_0$ is their \emph{difference}: two large nearly
parallel vectors whose difference can point in an uncontrolled direction
if their magnitudes nearly cancel. Near-cancellation is the finite-time
manifestation of the slow (exceptional) directions of $P_T$: it says
$\xi_0$ has a large component there. The two notions, near-cancellation
of the positive parts and largeness along the smallest-growth
directions of the finite product, are related but not identical without
an argument (Euclidean smallest-growth directions are singular-vector
subspaces, while cancellation is measured by a dual positive functional
from the approximate rank-one structure); making the equivalence
uniform in $T$ is precisely what the exponential-separation theory below
provides. At any fixed $T$ it is a condition on the finite product
$P_T$ alone, and the slow set drifts with $T$.

\textbf{Asymptotic refinement: a canonical invariant bundle.} The
finite-time statement never singles out a limiting object. If one wants
a genuinely $\Phi$-invariant \emph{limiting} bundle $E_k(t)$ that the
sliver converges to as $T\to\infty$, more is required, and it is worth
stating what, and being careful about which space the separation
theorem acts in. Suppose the forward orbit closure
$K = \overline{\{x(t):t\ge0\}}$ is compact and \emph{fully} invariant
under the time-$h_0$ map, so that $g := \varphi_{h_0}|_K$ is a
homeomorphism of $K$. (This is the hypothesis actually used; a
convenient stronger sufficient setting is that the reference trajectory
lie in a compact invariant set on which the flow is defined for all
$t\in\mathbb R$. A time-$h_0$ map of an ODE is injective wherever the
flow is defined, so $g$ is a continuous bijection of a compact space
onto itself, hence a homeomorphism.) The standing uniform-positivity
hypotheses were stated along the reference trajectory; they extend to
all of $K$ as follows. Since $t\mapsto Df(x(t))^{[k]}$ is continuous,
the inequalities assumed for a.e.\ $t$ in fact hold for every $t$; and
since $x\mapsto Df(x)^{[k]}$ is continuous with the orbit dense in its
closure, the fixed-edge inequalities ($B_{ij}(x)\ge\delta$ on edges
of $G$, $B_{ij}(x)\ge0$ off-diagonal, $|B_{ij}(x)|\le\Lambda$) carry to
every $x\in K$. Full invariance then makes every length-$h_0$ compound
fibre map over $K$ strongly positive with the same uniform estimates.

Set $Y := \wedge^k\mathbb R^n$, let $C\subset Y$ be the standard
positive compound cone, and let $\mathcal T_x := \wedge^k
D\varphi_{h_0}(x)$ be the compound window map. We use the
product-bundle version of Pol\'a\v{c}ik--Tere\v{s}\v{c}\'ak Theorem~1,
with the proof corrected by its Erratum. The theorem is stated for a
compact base $K$ and a homeomorphism $g:K\to K$, with a continuous
family of fibre maps that are strongly positive and compact as
operators (their theorem is set in possibly infinite-dimensional
ordered Banach spaces, where operator compactness, sending bounded sets
to relatively compact ones, is a genuine hypothesis). These hypotheses
hold here: every linear operator on the finite-dimensional space
$\wedge^k\mathbb R^n$ is compact in that sense, continuity follows from
continuous dependence of the variational flow
on the base point, and the finite-window positivity estimate makes each
$\mathcal T_x$ strongly positive ($\mathcal T_x(C\setminus\{0\})\subset
\operatorname{int}C$). The theorem yields, simultaneously, continuous
principal lines in the fibre and in its dual,
\[
   E^{\mathrm c}_k(x) = \operatorname{span}\{v_x\}\subset Y,
   \qquad
   E^{\mathrm c,*}_k(x) = \operatorname{span}\{\ell_x\}\subset Y^*,
\]
with $v_x\in\operatorname{int}C$ and $\ell_x$ strictly positive,
satisfying the covariance relations $\mathcal T_x E^{\mathrm c}_k(x) =
E^{\mathrm c}_k(g(x))$ and $\mathcal T_x^{*}E^{\mathrm c,*}_k(g(x)) =
E^{\mathrm c,*}_k(x)$. Defining $F^{\mathrm c}_k(x) := \ker\ell_x$ gives
a continuous, codimension-one complement containing no nonzero positive
vector; for the invertible finite-dimensional maps used here the
forward invariance of the theorem's general setting upgrades to
equality, $\mathcal T_x F^{\mathrm c}_k(x) = F^{\mathrm c}_k(g(x))$.
The discrete exponential-separation estimate is, for every integer
$m\ge1$ and every unit $w\in F^{\mathrm c}_k(x)$,
\[
   \big\lVert \mathcal T^{(m)}_x\,w\big\rVert
   \;\le\; M\gamma^{\,m}\,\big\lVert \mathcal T^{(m)}_x\,v_x\big\rVert,
   \qquad M>0,\ \gamma\in(0,1),
\]
where $\mathcal T^{(m)}_x = \mathcal T_{g^{m-1}(x)}\cdots\mathcal
T_{g(x)}\mathcal T_x$; since $E^{\mathrm c}_k$ is one-dimensional and
$v_x$ is normalised, this is exactly $\lVert \mathcal T^{(m)}_x|_{F^{\mathrm
c}_k}\rVert / \underline m(\mathcal T^{(m)}_x|_{E^{\mathrm c}_k}) \le
M\gamma^{\,m}$, with $\underline m(\cdot)$ the conorm.

This is a splitting \emph{of the compound space}
$Y=\wedge^k\mathbb R^n$; it is not yet the desired splitting of
$\mathbb R^n$. Descending uses the compound structure:
\begin{enumerate}
\item[(a)] the principal compound line $E^{\mathrm c}_k(x)$ is
decomposable (since $\wedge^k\Phi$ preserves decomposability, as in
Section~2), so it equals $\wedge^k E_k(x)$ for a genuine $k$-plane
$E_k(x)\subset\mathbb R^n$;
\item[(b)] the strictly positive principal covector $\ell_x$ is
likewise decomposable, $\ell_x = \ell_1(x)\wedge\cdots\wedge\ell_k(x)$.
This is not a consequence of positivity alone (strict positivity of a
$k$-covector does not force decomposability when $2\le k\le n-2$); it
follows because the adjoint fibre maps have the compound form
$\mathcal T_x^{*} = \wedge^k\big(D\varphi_{h_0}(x)^{*}\big)$ and hence
preserve decomposable covectors, so the principal dual ray is obtained
as a projective limit of positive decomposable covectors and the
Pl\"ucker image is closed, exactly as for the primal line in Section~2;
\item[(c)] set $F_k(x) := \bigcap_{i=1}^k \ker\ell_i(x)$, of dimension
$n-k$. Invariance comes from the covariance of the decomposable dual
line: $(\wedge^k A_x)^{*}\,\ell(g(x)) = c(x)\,\ell(x)$ with $c(x)>0$
says that $A_x^{*}$ carries the span of the target covectors
$\ell_i(g(x))$ onto the span of the source covectors $\ell_i(x)$, so
taking common kernels gives $A_x F_k(x) = F_k(g(x))$. Complementarity
$E_k(x)\oplus F_k(x) = \mathbb R^n$ holds because the pairing of the
strictly positive principal $k$-vector with the strictly positive
principal $k$-covector is nonzero;
\item[(d)] for a decomposable $\xi$ representing a plane $\Pi$, lying in
$F^{\mathrm c}_k(x) = \ker\ell_x$ is equivalent to $\Pi\cap F_k(x)
\ne\{0\}$, which is the Schubert exceptional set of the earlier
sections.
\end{enumerate}
These steps build the two bundles, their invariance, and the
exceptional set, but not yet the \emph{quantitative} domination
downstairs. That last step is a short singular-value estimate. Write
$A := A^{(m)}_t := \Phi(t+mh_0,\,t)$ for the downstairs iterate, and let
$s_1\ge s_2\ge\cdots\ge s_k>0$ be the singular values of $A|_{E_k(t)}$.
Choose $e_1,\dots,e_{k-1}$ to be the right singular vectors of
$A|_{E_k(t)}$ for $s_1,\dots,s_{k-1}$ (\emph{not} an arbitrary
orthonormal frame; this choice is what makes the product bound below
hold), and take a unit $f\in F_k(t)$ nearly maximising $\lVert
A|_{F_k(t)}\rVert$. The decomposable $e_1\wedge\cdots\wedge e_{k-1}\wedge
f$ lies in the compound slow hyperplane $F^{\mathrm c}_k(t)$. Here the
lower bound comes from the angle between the \emph{target} bundles:
since $A e_i \in E_k(t+mh_0)$ and $A f\in F_k(t+mh_0)$, the uniform
angle bound $\inf_x\angle(E_k(x),F_k(x))>0$ gives
$\operatorname{dist}\big(A f,\,E_k(t+mh_0)\big)\ge c_0\lVert A
f\rVert$, and therefore
\[
   \big\lVert \wedge^k A\,(e_1\wedge\cdots\wedge e_{k-1}\wedge f)\big\rVert
   \;\ge\; c_0\, s_1\cdots s_{k-1}\,\lVert A f\rVert,
   \qquad c_0>0 \text{ uniform.}
\]
On the principal compound line, meanwhile,
$\underline m\big(\wedge^k A\,\big|_{\wedge^k E_k(t)}\big) = s_1\cdots
s_k$. Dividing, and using the compound exponential-separation estimate
for the ratio, gives
\[
   \frac{\lVert A f\rVert}{s_k}
   \;\le\; C\,
   \frac{\big\lVert \wedge^k A\,\big|_{F^{\mathrm c}_k(t)}\big\rVert}
        {\underline m\big(\wedge^k A\,\big|_{E^{\mathrm c}_k(t)}\big)}
   \;\le\; C M\gamma^{\,m},
\]
that is, $\lVert A|_{F_k(t)}\rVert / \underline m(A|_{E_k(t)}) \le
C M\gamma^{\,m}$, which is the downstairs dominated splitting
$\mathbb R^n = E_k(t)\oplus F_k(t)$. The compactness of the base $K$ is
used at several points. It is not needed for operator-compactness of
the fibre maps (which is automatic here, every operator on the
finite-dimensional fibre being compact); rather, it provides
a compact base for the bundle theorem, the uniform positivity and
continuity estimates, and, from continuity of the bundle over a compact
base, the uniform angle bound $\inf_x\angle(E_k(x),F_k(x))>0$ (itself
used in the estimate just given). On a noncompact forward orbit, or
without full invariance, one would have to assume such a splitting
directly. The precise hypotheses of Pol\'a\v{c}ik--Tere\v{s}\v{c}\'ak
Theorem~1, verified against the source for the compound-cocycle
application, are exactly those met here: a compact base with a
homeomorphism base map, a continuous family of compact strongly
positive fibre maps, and the discrete integer-iterate separation
estimate used above; the printed theorem is stated for a product bundle
$K\times X$, and Remark~(i) there notes the routine extension to a
bundle with a fixed ordered fibre, which is the (simpler) situation
$K\times Y$ used here. For the ordered-space separation theorem see
Pol\'a\v{c}ik--Tere\v{s}\v{c}\'ak \cite{PolacikTerescak1993} (Theorem~1,
with the proof corrected by the Erratum \cite{PolacikTerescak1994}),
Fusco and Oliva \cite{FuscoOliva1991} for the autonomous single-operator
$k$-dimensional Perron theorem, and Birkhoff \cite{Birkhoff1957, Birkhoff1967} and
Hopf \cite{Hopf1963} for the underlying projective-metric contraction.

\textbf{On the frozen anchors $L(\tau)$ and $\Pi_{\mathrm{anc}}(\tau)$.} The devices
built in the two sections from the frozen matrices $Df(x_0)$ and
$Df(x_0)^{[2]}$ are, respectively, the transported Perron eigen-line and
eigen-plane of a \emph{single} matrix at $x_0$. They are exact and
elementary, which is why they are convenient for drawing and for
intuition. For the finite-time statement,
Theorem~\ref{thm:align}, they serve only as in-cone reference markers,
and any such marker would do. For the asymptotic refinement, they are
not in general the invariant bundle: the Perron data of the frozen
generator has no a priori reason to coincide with the limiting bundle of
the nonautonomous cocycle. When the frozen eigen-object is transverse to
the slow directions (the generic case, but an assumption), the angle
between the transported anchor and the limiting bundle tends to $0$
exponentially, so $L(\tau)$ and $\Pi_{\mathrm{anc}}(\tau)$ approach the
respective limiting bundles without ever equalling them at finite
$\tau$. The frozen eigen-plane is, at least, always a genuine 2-plane,
in every dimension, by the compound-decomposability argument of
Section~2, so the picture is faithful on that score.

\textbf{From tangent directions to nonlinear trajectories.}
It is worth being explicit about one gap.
Theorem~\ref{thm:align} and its refinement are about $\Phi(t,0)$, the
derivative cocycle; to turn them into statements about the actual
nonlinear map $\varphi_t$ one writes
\[
   \varphi_t(x_0+v_0) - \varphi_t(x_0) = \Phi(t,0)v_0 + R_t(v_0),
   \qquad R_t(v_0) = o(\lvert v_0\rvert)\ \text{for fixed } t,
\]
but the remainder $R_t$ is \emph{not} uniform in $t$: under expansion it
can grow with $t$, so the two limits $v_0\to0$ and $t\to\infty$ do not
commute. Consequently the ``landing on the reference orbit'' (Section~1)
and ``landing on the envelope'' (Section~2) pictures are rigorous only
in the regime $v_0\to0$ at fixed $t$, with the shift kept small; they
are first-order, fixed-time shadows of the directional
alignment, not standalone long-time nonlinear theorems. A genuine
long-time nonlinear statement requires shrinking $\lvert v_0\rvert$
suitably as $t$ grows, or additional dynamical hypotheses that yield
uniform control of the remainder, such as appropriate stability, normal
hyperbolicity, shadowing, or explicit estimates on $D^2\varphi_t$
relative to $D\varphi_t$. It is worth stressing that a compact invariant
set and a bound on $D^2 f$ do \emph{not} by themselves suffice: both
$D\varphi_t$ and $D^2\varphi_t$ may grow exponentially along a
trajectory in a compact set, so the Taylor remainder of $\varphi_t$ need
not be uniformly controlled. Such a statement is not claimed here.

\appendix

\section{Explicit entry bound for the window transition matrix}

This appendix supplies the inequality quoted in the remark ``The
Birkhoff rate, precisely'': it makes the lower bound of order
$(\delta h_0)^r/r!$ explicit, and derives from it the finite,
$t$-uniform bound on the projective diameter.

\textbf{Setup.} Write $B(t) := Df(x(t))^{[k]}$, an $N\times N$ matrix
with $N = \binom{n}{k}$, and let $\Psi(t_2,t_1)$ be the transition matrix
of the induced compound system $\dot\eta = B(t)\eta$, so that
$\partial_{t_2}\Psi(t_2,t_1) = B(t_2)\Psi(t_2,t_1)$, $\Psi(t_1,t_1) =
I$, and $\wedge^k D\varphi_\tau(x_0) = \Psi(\tau,0)$. We use the
standing hypotheses of uniform $k$-positivity from the introduction:
for a.e.\ $t$, $B(t)$ is Metzler ($B_{ij}(t) \ge 0$ for $i \ne j$); the
off-diagonal entries corresponding to the edges of a fixed strongly
connected directed graph $G$ on $\{1,\dots,N\}$ satisfy $B_{ij}(t) \ge
\delta > 0$; and all entries are bounded, $|B_{ij}(t)| \le
\Lambda$. Let $r$ be the diameter of $G$, the largest over ordered
pairs $(i,j)$ of the shortest directed path length; since $G$ is
strongly connected on $N$ nodes, $r \le N-1$. Fix a window length $h_0
> 0$ with $\delta h_0 \le 1$, and put $\mu := \Lambda$ and $M := \mu +
\Lambda = 2\Lambda$.

\textbf{Lemma.} For every $t$ and all $i,j$,
\[
   \beta \;\le\; \Psi(t+h_0,t)_{ij} \;\le\; \Gamma,
   \qquad
   \beta := e^{-\mu h_0}\,\frac{(\delta h_0)^{r}}{r!},
   \qquad
   \Gamma := 2\,e^{-\mu h_0}\,e^{h_0 N M}.
\]
Hence $\Psi(t+h_0,t)$ is entrywise strictly positive, and its Hilbert
projective diameter satisfies
\[
   \operatorname{diam}\Psi(t+h_0,t)
   \;\le\; 2\log\frac{\Gamma}{\beta}
   \;=\; 2\log\!\Big(2\,e^{h_0 N M}\,\frac{r!}{(\delta h_0)^{r}}\Big)
   \;=:\; D_0 < \infty,
\]
uniformly in $t$, so $\kappa_0 := \tanh(D_0/4) < 1$ is a valid
per-window Birkhoff contraction bound (the exact coefficient
$\tanh(\operatorname{diam}\Psi/4)$ is no larger, since $D_0$ bounds the
diameter and $\tanh$ is increasing).

\textbf{Proof.} Write $B(t) = -\mu I + Q(t)$ with $Q(t) := B(t) + \mu
I$. As $\mu = \Lambda \ge -B_{ii}(t)$, the diagonal of $Q$ is
nonnegative and its off-diagonal agrees with that of $B$, so $Q(t) \ge
0$ entrywise with entries at most $\mu + \Lambda = M$. The substitution
$\eta = e^{-\mu(t-t_1)}\zeta$ turns $\dot\eta = B\eta$ into $\dot\zeta =
Q\zeta$, giving $\Psi(t_2,t_1) = e^{-\mu(t_2-t_1)}\Phi(t_2,t_1)$, where
$\Phi$ is the transition matrix of $\dot\zeta = Q(t)\zeta$, expanded in
the Peano--Baker series
\[
   \Phi(t_2,t_1) = \sum_{m=0}^{\infty}
   \int_{t_1}^{t_2}\!\!\int_{t_1}^{s_1}\!\!\cdots\!\int_{t_1}^{s_{m-1}}
   Q(s_1)Q(s_2)\cdots Q(s_m)\; ds_m\cdots ds_1 .
\]
Each factor $Q(s)$ is entrywise nonnegative, so every term is entrywise
nonnegative and $\Phi$ dominates any single term.

\emph{Lower bound.} Fix $(i,j)$. If $i = j$, the $m=0$ term is $I$, so
$\Phi(t_2,t_1)_{ii} \ge 1$. If $i \ne j$, note that for the column-vector
dynamics $\dot\eta = Q(t)\eta$ the entry $Q_{ab}$ represents influence
of coordinate $b$ on coordinate $a$, so $[Q(s_1)\cdots Q(s_m)]_{ij}$
sums over directed paths from $j$ to $i$. Choose a shortest such path
$j = i_m \to i_{m-1} \to \cdots \to i_0 = i$ in $G$ (strong connectivity
provides one, of length $m \le r$), with every chosen edge belonging to
$G$ and hence bounded below by $\delta$, so
$Q_{i_{\ell-1} i_\ell}(s) = B_{i_{\ell-1} i_\ell}(s) \ge \delta$ for
a.e.\ $s$. Keeping in the $m$-th term only this single path,
\[
   [\,Q(s_1)\cdots Q(s_m)\,]_{ij}
   \;\ge\; Q_{i_0 i_1}(s_1)\,Q_{i_1 i_2}(s_2)\cdots Q_{i_{m-1} i_m}(s_m)
   \;\ge\; \delta^m ,
\]
and integrating over the ordered simplex, whose volume is
$(t_2-t_1)^m/m!$,
\[
   \Phi(t_2,t_1)_{ij} \;\ge\; \delta^m\,\frac{(t_2-t_1)^m}{m!}.
\]
With $t_2 - t_1 = h_0$ and $\delta h_0 \le 1$, the sequence $m \mapsto
(\delta h_0)^m/m!$ is nonincreasing (successive ratio $\delta
h_0/(m+1) \le 1$), so for $1 \le m \le r$ it is at least $(\delta
h_0)^r/r!$. Thus, for $i \ne j$,
\[
   \Psi(t+h_0,t)_{ij} = e^{-\mu h_0}\Phi(t+h_0,t)_{ij}
   \;\ge\; e^{-\mu h_0}\,\frac{(\delta h_0)^{r}}{r!} = \beta ,
\]
while for $i=j$, $\Psi(t+h_0,t)_{ii} \ge e^{-\mu h_0} \ge \beta$ (since
$(\delta h_0)^r/r! \le 1$). This is the lower bound.

\emph{Upper bound.} A product of $m$ matrices with entries in $[0,M]$
has entries at most $N^{m-1}M^m$: each entry of the product is a sum
over the $N^{m-1}$ choices of intermediate indices $i_1,\dots,i_{m-1}$,
each summand a product of $m$ entries bounded by $M$. So the $(i,j)$
entry of the $m$-th
term is at most $N^{m-1}M^m (t_2-t_1)^m/m!$. The $m=0$ term contributes
at most $1$. Summing,
\[
   \Phi(t+h_0,t)_{ij}
   \;\le\; 1 + \sum_{m=1}^\infty \frac{N^{m-1}M^m h_0^m}{m!}
   \;\le\; 1 + \tfrac1N e^{h_0 N M}
   \;\le\; 2\,e^{h_0 N M},
\]
so $\Psi(t+h_0,t)_{ij} = e^{-\mu h_0}\Phi(t+h_0,t)_{ij} \le \Gamma$.

\emph{Diameter.} For any entrywise positive $P$ with all entries in
$[\beta,\Gamma]$,
\[
   \operatorname{diam}(P) = \max_{i,j,p,q}\log\frac{P_{ip}P_{jq}}{P_{iq}P_{jp}}
   \;\le\; \log\frac{\Gamma\cdot\Gamma}{\beta\cdot\beta}
   = 2\log\frac{\Gamma}{\beta}.
\]
Applied to $P = \Psi(t+h_0,t)$ this gives $D_0$, uniform in $t$, and
by Birkhoff's theorem the per-window contraction ratio is
$\tanh(\operatorname{diam}\Psi/4) \le \tanh(D_0/4) =: \kappa_0 < 1$, so
$\kappa_0$ is a valid uniform contraction bound. $\blacksquare$

\textbf{Composition.} Partition $[0,\tau]$ into $\lfloor \tau/h_0\rfloor$
windows of length $h_0$ (with a shorter remainder, which is
non-expansive on the cone where the Hilbert metric is defined and so
only improves the bound). Since $\wedge^k
D\varphi_\tau(x_0) = \Psi(\tau,0)$ factors as the product of the
corresponding window maps, each contracting the Hilbert metric by
$\kappa_0$,
\[
   d_H\!\big(\wedge^k D\varphi_\tau\,\xi,\ \wedge^k D\varphi_\tau\,\eta\big)
   \;\le\; \kappa_0^{\lfloor \tau/h_0\rfloor}\, d_H(\xi,\eta)
   \;\le\; \kappa_0^{-1}\,e^{-c\tau}\, d_H(\xi,\eta),
\]
that is, exponential decay at rate $c = -h_0^{-1}\log\kappa_0 >
0$ (with prefactor $\kappa_0^{-1}$, not $1$), the estimate used in the
main text. The dependence on the band and
the connectivity is now visible in closed form: $\beta$ carries the
factor $(\delta h_0)^r/r!$, so the rate degrades as the lower band
$\delta$ shrinks or the graph diameter $r$ grows, and improves as they
sharpen.

\section*{Acknowledgment}

In preparing these notes I made use of AI tools for drafting,
literature cross-checking, and iterative editing; the mathematical
content, viewpoint, and any errors are my own.

\end{document}